\documentclass[11pt]{article}
\usepackage{graphicx}
\usepackage{tabularx}
\usepackage{url}
\usepackage{amsthm}
\usepackage{amsmath}
\usepackage{amsfonts}
\usepackage{amssymb}

\theoremstyle{plain}
\newtheorem{theorem}{Theorem}
\newtheorem{lemma}{Lemma}

\theoremstyle{definition}

\theoremstyle{remark}

\date{}

\begin{document}

\title{Decision trees for regular factorial languages}
\author{Mikhail Moshkov\thanks{Computer, Electrical and Mathematical Sciences and Engineering Division,
King Abdullah University of Science and Technology (KAUST),
Thuwal 23955-6900, Saudi Arabia. Email: mikhail.moshkov@kaust.edu.sa.
}}
\maketitle

\begin{abstract}
In this paper, we study arbitrary regular factorial languages over a finite
alphabet $\Sigma$. For the set of words $L(n)$ of the length $n$ belonging to a
regular factorial language $L$, we investigate the depth of decision trees
solving the recognition and the membership problems deterministically and
nondeterministically. In the case of recognition problem, for a given word
from $L(n)$, we should recognize it using queries each of which, for some $%
i\in \{1,\ldots ,n\}$, returns the $i$th letter of the word. In the case of
membership problem, for a given word over the alphabet $\Sigma$ of the length 
$n$, we should recognize if it belongs to the set $L(n)$ using the same
queries. For a given problem and type of trees, instead of the minimum depth 
$h(n)$ of a decision tree of the considered type solving the problem for  
$L(n)$, we study the smoothed minimum depth  $H(n)=\max\{h(m):m\le n\}$.
With the growth of $n$, the smoothed minimum depth of decision trees solving
the problem of recognition deterministically is either bounded from above by
a constant, or grows as a logarithm, or linearly.
For other cases (decision trees solving the problem of recognition
nondeterministically, and decision trees solving the membership problem
deterministically and nondeterministically), with the growth of $n$, the smoothed
minimum depth of decision trees is either bounded from above by a constant
or grows linearly. As corollaries of the obtained results, we study joint 
behavior of smoothed minimum depths of decision trees for the considered 
four cases and describe five complexity classes of regular factorial languages. 
We also investigate  the class of regular factorial languages over the alphabet 
$\{0,1\}$ each of which is given by one forbidden word.
\end{abstract}

\section{Introduction}
\label{sect1}

In this paper, we study arbitrary regular factorial languages over a finite
alphabet $\Sigma$.
For the set of words $L(n)$ of the length $n$ belonging to a regular factorial language $L$, we investigate the depth of decision trees
solving the recognition and the membership problems deterministically and
nondeterministically. In the case of recognition problem, for a given word
from $L(n)$, we should recognize it using queries each of which, for some $%
i\in \{1,\ldots ,n\}$, returns the $i$th letter of the word. In the case of
membership problem, for a given word over the alphabet $\Sigma$ of the length $n$%
, we should recognize if it belongs to $L(n)$ using the same queries.

 For a given problem (problem of recognition or membership problem) and type of trees (solving the problem deterministically or nondeterministically), instead of the minimum depth $h(n)$ of a decision tree of the considered type solving the problem for  $L(n)$, we study the smoothed minimum depth  $H(n)=\max\{h(m):m\le n\}$. The reason is that the graph of the function $h(n)$ may have sawtooth form.

For an arbitrary regular factorial language, with the growth of $n$, the smoot\-hed
minimum depth of decision trees solving the problem of recognition
deterministically is either bounded from above by a constant, or grows as a
logarithm, or linearly. These results follow immediately from more general, obtained in \cite{Moshkov00} for arbitrary regular languages.

For other cases (decision
trees solving the problem of recognition nondeterministically, and decision
trees solving the membership problem deterministically and
nondeterministically), with the growth of $n$, the smoothed minimum depth of decision
trees is either bounded from above by a constant, or grows linearly.
In the conference paper \cite{Moshkov97}, a classification of arbitrary regular languages depending on the smoothed minimum depth of decision trees solving the problem of recognition nondeterministically was announced without proofs. In the present paper, we consider simpler classification for regular factorial languages with full proof. Results related to the decision trees solving the membership problem are new.

As corollaries of the obtained results, we study joint behavior of smoothed minimum depths of decision trees for the considered four cases and describe five complexity classes of regular factorial languages. We also investigate  the class of regular factorial languages over the alphabet $E=\{0,1\}$ each of which is given by one forbidden word.

We should mention a recent paper \cite{Moshkov20} in which similar results were obtained for subword-closed languages over the alphabet  $E$, where by subword we mean subsequence. It is clear that each subword-closed language is a factorial language. Moreover, each subword-closed language over a finite alphabet is a regular language \cite{Haines67}. One can show that the language $L(00)$ over the alphabet $E$ given by one forbidden word $00$ is a regular factorial language, which is not subword-closed. Therefore the class of subword-closed languages over the alphabet  $E$ is a proper subclass of the class of regular factorial languages over the alphabet $E$.

The main difference between the present paper and \cite{Moshkov20} is that, in the latter paper,
we do not assume that the subword-closed languages are
given by sources (partial deterministic finite automata). Instead of this, we describe simple criteria (based on the presence in the language of words of special types) for
the behavior of the minimum depths of decision trees solving the problem of
recognition deterministically and nondeterministically. Differently formulated criteria for the behavior of the minimum depth of decision trees solving the recognition problem require very different proofs. One more difference is that in \cite{Moshkov20} we directly consider the minimum depth of decision trees since it is a fairly smooth function for subword-closed languages.

The rest of the paper is organized as follows. In Section \ref{sect2}, we
consider main notions, in Section \ref{sect3} -- main results, and in
Section \ref{sect4} -- two corollaries of these results.

\section{Main Notions}

\label{sect2}

In this section, we discuss the notions related to regular factorial
languages and decision trees solving problems of recognition and membership
for these languages.

\subsection{Regular Factorial Languages}

\label{sect2.1}

Let $\omega =\{0,1,2,\ldots \}$ be the set of nonnegative integers and $\Sigma$
be a finite alphabet with at least two letters. By $\Sigma^{\ast }$, we denote the set of all finite words
over the alphabet $\Sigma$, including the empty word $\lambda $. A word $w\in
\Sigma^{\ast}$ is called a factor of a word $u\in \Sigma^{\ast}$ if $u=v_{1}wv_{2}$
and $v_{1},v_{2}\in \Sigma^{\ast}$. A language $L\subseteq \Sigma^{\ast}$ is called
factorial if it contains all factors of its words. A word $w\in \Sigma^{\ast}$
is called a minimal forbidden word for $L$ if $w\notin L$ and all proper
factors of $w$ belong to $L$. We denote by $\mathit{MF}(L)$ the language of minimal
forbidden words for $L$. It is known \cite{Crochemore98} that a factorial language $L$ is
regular if and only if the language $\mathit{MF}(L)$ is regular. In particular, a
factorial language $L$ with a finite set of minimal forbidden words $\mathit{MF}(L)$
is regular. In this paper, we study arbitrary nonempty regular factorial
languages.

A source over the alphabet $\Sigma$ is a triple $I=(G,q_{0},Q)$, where $G$ is a finite
directed graph, possibly with multiple edges and loops, in which each edge
is labeled with a letter from $\Sigma$ and edges leaving each node are labeled
with pairwise different letters, $q_{0}$ is a node of $G$ called initial, and
$Q$ is a nonempty set of the graph $G$ nodes called terminal. Note that the source $I$
can be interpreted as a partial deterministic finite automaton.

A path of the source $I$ is an arbitrary sequence $\xi =v_{1},d_{1},\ldots
,v_{m},d_{m},v_{m+1}$ of nodes and edges of $G$ such that  the edge $d_{i}$ leaves the node $v_{i}$ and enters the node $v_{i+1}$ for $%
i=1,\ldots ,m$. We now define a word $w(\xi )$ from $\Sigma^{\ast }$ in the
following way: if $m=0$, then $w(\xi )=\lambda $. Let $m>0$ and let $\delta
_{j}$ be the letter attached to the edge $d_{j}$, $j=1,\ldots ,m$. Then $%
w(\xi )=\delta _{1}\cdots \delta _{m}$. We say that the path $\xi $
generates the word $w(\xi )$. Note that different paths which start in the
same node generate different words.

We denote by $\Xi (I)$ the set of all paths of the source $I$ each of which
starts in the node $q_{0}$ and finishes in a node from $Q$. Let
$$L_{I}=\{w(\xi ):\xi \in \Xi (I)\}.$$
We say that the source $I$ generates the language $L_{I}$. It is well known
that $L_{I}$ is a regular language.

The source $I$ is called everywhere defined over the alphabet $\Sigma$ if exactly
$|\Sigma|$ edges leave each node of $G$. Note that these edges are labeled with
pairwise different letters from $\Sigma$. The source $I$ is called reduced if,
for each node of $G$, there exists a path from $\Xi (I)$, which contains
this node. It is known \cite{Markov82} that, for each regular language over the alphabet
$\Sigma$, there exists an everywhere defined over the alphabet $\Sigma$ source, which
generates this language. Therefore, for each nonempty regular language,
there exists a reduced source, which generates this language.

Let $L$ be a regular factorial language and  $I=(G,q_{0},Q)$ be a reduced source that generates the language $L$. Since the language $L$ is factorial, we can assume additionally that each node of the graph $G$ is terminal -- it will not change the language generated by $I$. The source $I$
will be called t-reduced if it is reduced and each node of the graph $G$ is
terminal.
Further we will assume that a considered regular factorial language $L$ is
nonempty and it is given by a t-reduced source, which generates this
language.

\subsection{Decision Trees for Recognition and Membership Problems}

\label{sect2.2}

Let $L$ be a regular factorial language over the alphabet $\Sigma$. For any natural $n$, denote $L(n)=L\cap \Sigma^{n}$, where $\Sigma^{n}$ is
the set of words over the alphabet $\Sigma$, which length is equal to $n$. We consider
two problems related to the set $L(n)$. The problem of recognition: for a
given word from $L(n)$, we should recognize it using attributes (queries) $%
l_{1}^{n},\ldots ,l_{n}^{n}$, where $l_{i}^{n}$, $i\in \{1,\ldots ,n\}$, is
a function from $\Sigma^{n}$ to $\Sigma$ such that $l_{i}^{n}(a_{1}\cdots
a_{n})=a_{i}$ for any word $a_{1}\cdots a_{n}\in \Sigma^{n}$. The problem
of membership: for a given word from $\Sigma^{n}$, we should recognize if
this word belongs to the set $L(n)$ using the same attributes. To solve
these problems, we use decision trees over $L(n)$.

A decision tree over $L(n)$ is a marked finite directed tree with root,
which has the following properties:

\begin{itemize}
\item The root and the edges leaving the root are not labeled.

\item Each node, which is not the root nor terminal node, is labeled with an
attribute from the set $\{l_{1}^{n},\ldots ,l_{n}^{n}\}$.

\item Each edge leaving a node, which is not a root, is labeled with a
letter from the alphabet $\Sigma$.
\end{itemize}

A decision tree over $L(n)$ is called deterministic if it satisfies the
following conditions:

\begin{itemize}
\item Exactly one edge leaves the root.

\item For any node, which is not the root nor terminal node, the edges
leaving this node are labeled with pairwise different letters.
\end{itemize}

Let $\Gamma $ be a decision tree over $L(n)$. A complete path in $%
\Gamma $ is any sequence $\xi =v_{0},e_{0},\ldots ,v_{m},e_{m},v_{m+1}$ of
nodes and edges of $\Gamma $ such that $v_{0}$ is the root, $v_{m+1}$ is a
terminal node, and $v_{i}$ is the initial and $v_{i+1}$ is the terminal node
of the edge $e_{i}$ for $i=0,\ldots ,m$. We define a subset $\Sigma(n,\xi )$ of
the set $\Sigma^{n}$ in the following way: if $m=0$, then $\Sigma(n,\xi
)=\Sigma^{n}$. Let $m>0$, the attribute $l_{i_{j}}^{n}$ be attached to the
node $v_{j}$, and $b_{j}$ be the letter attached to the edge $e_{j}$, $%
j=1,\ldots ,m$. Then
\[
\Sigma (n,\xi )=\{a_{1}\cdots a_{n}\in \Sigma^{n}:a_{i_{1}}=b_{1},\ldots
,a_{i_{m}}=b_{m}\}.
\]

Let $L(n)\neq \emptyset $. We say that a decision tree $\Gamma $ over $L(n)$
solves the problem of recognition for $L(n)$ nondeterministically if $\Gamma
$ satisfies the following conditions:

\begin{itemize}
\item Each terminal node of $\Gamma $ is labeled with a word from $L(n)$.

\item For any word $w\in L(n)$, there exists a complete path $\xi $ in the
tree $\Gamma $ such that $w\in \Sigma (n,\xi )$.

\item For any word $w\in L(n)$ and for any complete path $\xi $ in the tree $%
\Gamma $ such that $w\in \Sigma (n,\xi )$, the terminal node of the path $\xi $ is
labeled with the word $w$.
\end{itemize}

We say that a decision tree $\Gamma $ over $L(n)$ solves the problem of
recognition for $L(n)$ deterministically if $\Gamma $ is a deterministic
decision tree, which solves the problem of recognition for $L(n)$
nondeterministically.

We say that a decision tree $\Gamma $ over $L(n)$ solves the problem of
membership for $L(n)$ nondeterministically if $\Gamma $ satisfies the
following conditions:

\begin{itemize}
\item Each terminal node of $\Gamma $ is labeled with a number from the set $\{0,1\}$.

\item For any word $w\in \Sigma^{n}$, there exists a complete path $\xi $
in the tree $\Gamma $ such that $w\in \Sigma (n,\xi )$.

\item For any word $w\in\Sigma^{n}$ and for any complete path $\xi $ in
the tree $\Gamma $ such that $w\in \Sigma (n,\xi )$, the terminal node of the path
$\xi $ is labeled with the number $1$ if $w\in L(n)$ and with the number $0$%
, otherwise.
\end{itemize}

We say that a decision tree $\Gamma $ over $L(n)$ solves the problem of
membership for $L(n)$ deterministically if $\Gamma $ is a deterministic
decision tree which solves the problem of membership for $L(n)$
nondeterministically.

Let $\Gamma $ be a decision tree over $L(n)$. We denote by $h(\Gamma )$ the
maximum number of nodes in a complete path in $\Gamma $ that are not the
root nor terminal node. The value $h(\Gamma )$ is called the depth of the
decision tree $\Gamma $.

We denote by $h_{L}^{ra}(n)$ ($h_{L}^{rd}(n)$) the minimum depth of a
decision tree over $L(n)$, which solves the problem of recognition for $L(n)$
nondeterministically (deterministically). If $L(n)=\emptyset $, then $%
h_{L}^{ra}(n)=h_{L}^{rd}(n)=0$.

We denote by $h_{L}^{ma}(n)$ ($h_{L}^{md}(n)$) the minimum depth of a
decision tree over $L(n)$, which solves the problem of membership for $L(n)$
nondeterministically (deterministically). If $L(n)=\emptyset $, then $%
h_{L}^{ma}(n)=h_{L}^{md}(n)=0$.

\section{Bounds on Decision Tree Depth}
\label{sect3}

Let $L$ be a nonempty factorial regular language. In this section, we
consider the behavior of four functions $H_{L}^{ra}$, $H_{L}^{rd}$, $%
H_{L}^{ma}$, and $H_{L}^{md}$ defined on the set $\omega \setminus \{0\}$
and with values from $\omega $. For any natural $n$,%
\begin{eqnarray*}
H_{L}^{ra}(n) &=&\max \{h_{L}^{ra}(m):1\leq m\leq n\}, \\
H_{L}^{rd}(n) &=&\max \{h_{L}^{rd}(m):1\leq m\leq n\}, \\
H_{L}^{ma}(n) &=&\max \{h_{L}^{ma}(m):1\leq m\leq n\}, \\
H_{L}^{md}(n) &=&\max \{h_{L}^{md}(m):1\leq m\leq n\}.
\end{eqnarray*}%
For any pair $bc\in \{ra,rd,ma,md\}$, the function $%
H_{L}^{bc}(n)$ is a smoothed analog of the function $h_{L}^{bc}(n)$.

\subsection{Decision Trees Solving Recognition Problem Deterministically}
\label{sect3.1}

Let $I=(G,q_{0},Q)$ be a t-reduced source over the alphabet $\Sigma$.  A path of the source $I$ is  called a cycle of the source $I$
if there is at least one edge in this path, and the first node of this path
is equal to the last node of this path. A cycle of the source $I$ is called
elementary if nodes of this cycle, with the exception of the last node, are
pairwise different.

The source $I$ is called simple if every two different
elementary cycles of the source $I$ do not have common nodes. Let $I$ be a
simple source and $\xi $ be a path of the source $I$. The number of
different elementary cycles of the source $I$, which have
common nodes with $\xi $, is denoted by $cl(\xi )$ and is called the cyclic
length of the path $\xi $. The value
$$
cl(I)=\max \{cl(\xi ):\xi \in \Xi (I)\}
$$
is called the cyclic length of the source $I$.

Let $I$ be a simple source, $C$ be an elementary cycle of the source $I$,
and $v$ be a node of the cycle $C$. Beginning with the node $v$, the cycle $C$
generates an infinite periodic word over the alphabet $\Sigma$. This word will be
denoted by $W(I,C,v)$. We denote by $r(I,C,v)$ the minimum period of the
word $W(I,C,v)$. The source $I$ is called dependent if there exist two
different elementary cycles $C_{1}$ and $C_{2}$ of the source
$I$, nodes $v_{1}$ and $v_{2}$ of the cycles $C_{1}$ and $C_{2}$,
respectively, and a path $\pi $ of the source $I$ from $v_{1}$ to $v_{2}$,
which satisfy the following conditions: $W(I,C_{1},v_{1})=W(I,C_{2},v_{2})$
and the length of the path $\pi $ is a number divisible by $r(I,C_{1},v_{1})$%
. If the source $I$ is not dependent, then it is called independent.
Next theorem follows immediately from Theorem 2.1  \cite{Moshkov00}.

\begin{theorem}
\label{T1} Let $L$ be a nonempty regular factorial language over the alphabet $\Sigma$ and $I$ be a
t-redu\-ced source, which generates the language $L$. Then the following
statements hold:

(a) If $I$ is an independent simple source and $cl(I)\leq 1$, then $%
H_{L}^{rd}(n)=O(1)$.

(b) If $I$ is an independent simple source and $cl(I)\geq 2$, then $%
H_{L}^{rd}(n)=\Theta (\log n)$.

(c) If $I$ is not independent simple source, then $H_{L}^{rd}(n)=\Theta (n)$.
\end{theorem}

\subsection{Decision Trees Solving Recognition Problem Nondeterministically}

\label{sect3.2}

Let $L$ be a nonempty regular factorial language over the alphabet $\Sigma$. For
any natural $n$, we define a parameter $T_{L}(n)$ of the language $L$. If $%
L(n)=\emptyset $, then $T_{L}(n)=0$. Let $L(n)\neq \emptyset $,  $%
w=a_{1}\cdots a_{n}\in L(n)$, and $J\subseteq \{1,\ldots ,n\}$.  Denote $%
L(w,J)=\{b_{1}\cdots b_{n}\in L(n):b_{j}=a_{j},j\in J\}$ (if $J=\emptyset $,
then $L(w,J)=L(n)$) and $M_{L}(n,w)=\min \{|J|:J\subseteq \{1,\ldots
,n\},|L(w,J)|=1\}$. Then $$T_{L}(n)=\max \{M_{L}(n,w):w\in L(n)\}.$$ Note
that, for any word $w\in L(n)$, $M_{L}(n,w)$ is the minimum number of
letters of the word $w$, which allow us to distinguish it from all
other words belonging to $L(n)$. One can show that $h_{L}^{ra}(n)=T_{L}(n)$.

\begin{theorem}
\label{T2} Let $L$ be a nonempty regular factorial language over the alphabet $\Sigma$ and $%
I=(G,q_{0},Q)$ be a t-reduced source, which generates the language $L$. Then
the following statements hold:

(a) If $I$ is an independent simple source, then $%
H_{L}^{ra}(n)=O(1)$.

(b) If $I$ is not independent simple source, then $H_{L}^{ra}(n)=\Theta (n)$.
\end{theorem}

\begin{proof}
(a) Let $I$ be an independent simple source and $cl(I)\leq 1$. By Theorem %
\ref{T1}, $H_{L}^{rd}(n)=O(1)$. It is clear that $H_{L}^{ra}(n)\leq
H_{L}^{rd}(n)$. Therefore $H_{L}^{ra}(n)=O(1)$.

Let $I$ be an independent simple source and $cl(I)\geq 2$. Let $n$ be a
natural number. If $L(n)=\emptyset $, then $T_{L}(n)=0$. Let $L(n)\neq
\emptyset $. Denote by $d$ the number of nodes in the graph $G$. In the
proof of Lemma 4.5 \cite{Moshkov00}, it was proved that $M_{L}(n,w)\leq d(4d+1)$ for any word
$w\in L(n)$. Therefore $T_{L}(n)\leq d(4d+1)$. Thus, $h_{L}^{ra}(n)\leq
d(4d+1)$ for any natural $n$ and $H_{L}^{ra}(n)=O(1)$.

(b) Let $I$ be not simple source and $C_{1},C_{2}$ be different elementary cycles of the source $I$, which have a common node $v$%
. Since $I$ is a t-reduced source, it contains a path $\xi $ from the node $%
q_{0}$ to the node $v$, and $v$ is a terminal node. Let the length of the path $\xi $ be
equal to $a$, the length of the cycle $C_{1}$ be equal to $b$, and the length
of the cycle $C_{2}$ be equal to $c$. Let $\alpha $ be the word generated by
the  path $\xi $, $\beta $ be the word generated by a path from $v$
to $v$ obtained by the passage $c$ times along the cycle $C_{1}$, and $%
\gamma $ be the word generated by a path from $v$ to $v$ obtained by the
passage $b$ times along the cycle $C_{2}$. The words $\beta $ and $\gamma $
are different and they have the same length $bc$.

Consider the sequence of
numbers $n_{i}=a+ibc$, $i=1,2,\ldots $. Let $i \in \omega \setminus \{0\}$. The set $L(n_{i})
$ contains the word $\alpha \gamma ^{i}$ and the words $\alpha \gamma
^{j}\beta \gamma ^{i-j-1}$ for $j=0,\ldots ,i-1$. It is easy to show that $%
M_{L}(n,\alpha \gamma ^{i})\geq i$: to distinguish the word $\alpha \gamma
^{i}$ from the words $\alpha \gamma ^{j}\beta \gamma ^{i-j-1}$, $j=0,\ldots
,i-1$, we need to use at least one letter from each of $i$ words $\gamma $
appearing in $\alpha \gamma ^{i}$. Therefore $T_{L}(n_{i})\geq i$ and $%
h_{L}^{ra}(n_{i})\geq i=(n_{i}-a)/(bc)$. Let $n\geq n_{1}$ and let $i$ be
the maximum natural number such that $n\geq n_{i}$. Evidently, $%
n-n_{i}\leq bc$. Hence $H_{L}^{ra}(n)\geq h_{L}^{ra}(n_{i})\geq (n-bc-a)/(bc)
$. Therefore $H_{L}^{ra}(n)\geq n/(2bc)$ for large enough $n$. The
inequality $H_{L}^{ra}(n)\leq n$ is obvious. Thus, $H_{L}^{ra}(n)=\Theta (n)$%
.

Let $I$ be a dependent simple source. Then there exist two different
elementary cycles $C_{1}$ and $C_{2}$ of the source $I$, nodes $v_{1}$ and $%
v_{2}$ of the cycles $C_{1}$ and $C_{2}$, respectively, and a path $\pi $
of the source $I$ from $v_{1}$ to $v_{2}$, which satisfy the following
conditions: $W(I,C_{1},v_{1})=W(I,C_{2},v_{2})$ and the length of the path $%
\pi $ is a number divisible by $r(I,C_{1},v_{1})$. Let us remind that, for $%
i=1,2$, $W(I,C_{i},v_{i})$ is the infinite periodic word over the alphabet $\Sigma
$ generated by the cycle $C_{i}$ beginning with the node $v_{i}$, and $%
r(I,C_{1},v_{1})$ is the minimum period of the word $W(I,C_{1},v_{1})$.
Since $I$ is a t-reduced source, it contains a path $\xi $ from the node $%
q_{0}$ to the node $v_{1}$, and all nodes of the graph $G$ are terminal. Let
the path $\xi $ generate the word $\alpha $ of the length $a$. Denote $r=$ $%
r(I,C_{1},v_{1})$. Let the length of the cycle $C_{1}$ be equal to $br$, the
length of the path $\pi $ be equal to $cr$, and the path $\pi $ generate the
word $\beta $. Denote by $\gamma $ the prefix of the length $r$ of the word $%
W(I,C_{1},v_{1})$. We now define two words of the length $rbc$: $u=\gamma
^{bc}$ and $w=\beta \gamma ^{c(b-1)}$. It is clear that $u\neq w$.

Consider
the sequence of numbers $n_{i}=a+irbc$, $i=1,2,\ldots $. Let $i\in \omega
\setminus \{0\}$. The set $L(n_{i})$ contains the word $\alpha u^{i}$ and
the words $\alpha u^{j}wu^{i-j-1}$ for $j=0,\ldots ,i-1$. It is easy to show
that $M_{L}(n,\alpha u^{i})\geq i$: to distinguish the word $\alpha u^{i}$
from the words $\alpha u^{j}wu^{i-j-1}$, $j=0,\ldots ,i-1$, we need to use
at least one letter from each of $i$ words $u$ appearing in $\alpha u^{i}$.
Therefore $T_{L}(n_{i})\geq i$ and $h_{L}^{ra}(n_{i})\geq i=(n_{i}-a)/(rbc)$%
. Let $n\geq n_{1}$ and let $i$ be the maximum natural number  such
that $n\geq n_{i}$. Evidently, $n-n_{i}\leq rbc$. Hence $H_{L}^{ra}(n)\geq
h_{L}^{ra}(n_{i})\geq (n-rbc-a)/(rbc)$. Therefore $H_{L}^{ra}(n)\geq n/(2rbc)
$ for large enough $n$. The inequality $H_{L}^{ra}(n)\leq n$ is obvious.
Thus, $H_{L}^{ra}(n)=\Theta (n)$. 
\end{proof}

Note that in general case (when we consider not only factorial languages)
the classification of reduced sources depending on the minimum depth of decision trees solving the problem of recognition nondeterministically is more complicated \cite{Moshkov97}. In
particular, there exists a dependent simple reduced source $I_0$ (see Fig. \ref{Fig0}) with the initial node labeled with the symbol $+$ and the unique terminal node labeled with the symbol $*$ that
generates the regular language $L_0=\{0^i10^j:i,j\in \omega\}$ over the alphabet $\{0,1\}$, which is not factorial and for which $H_{L_0}^{ra}(n)=O(1)$.

\begin{figure}[th]
\centering
\includegraphics[width=0.40\textwidth]{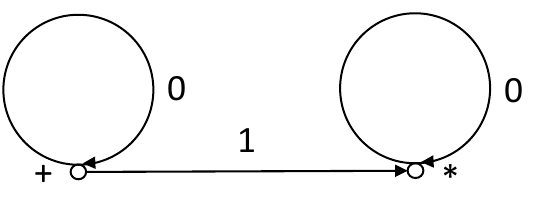}
\caption{Source $I_{0}$}
\label{Fig0}
\end{figure}

\subsection{Decision Trees Solving Membership Problem}

\label{sect3.3}

For a regular factorial language $L$ over the alphabet $\Sigma$, we denote by $L^{C}$ its complementary
language $\Sigma^{\ast }\setminus L$. The notation $|L|=\infty $ means that $L$
is an infinite language, and the notation $|L|<\infty $ means that $L$ is a
finite language.

\begin{theorem}
\label{T3} Let $L$ be a regular factorial language over the alphabet $\Sigma$.

(a) If $|L|=\infty $ and $L^{C}\neq \emptyset $, then $H_{L}^{md}(n)=\Theta
(n)$ and $H_{L}^{ma}(n)=\Theta (n)$.

(b) If $|L|<\infty $ or $L^{C}=\emptyset $, then $H_{L}^{md}(n)=O(1)$ and $%
H_{L}^{ma}(n)=O(1)$.
\end{theorem}

\begin{proof}
It is clear that $h_{L}^{ma}(n)\leq h_{L}^{md}(n)$ for any natural $n$.

(a) Let $|L|=\infty $, $L^{C}\neq \emptyset $, and $w_{0}$ be a word with
the minimum length from $L^{C}$. Denote by $t$ the length of $w_{0}$. Since $|L|=\infty $, $%
L(n)\neq \emptyset $ for any natural $n$. Let $n$ be a natural number such
that $n>t$ and $\Gamma $ be a decision tree over $L(n)$ that solves the
problem of membership for $L(n)$ nondeterministically and has the minimum
depth. Let $w\in L(n)$ and $\xi $ be a complete path in $\Gamma $ such that $%
w\in \Sigma (n,\xi )$. Then the terminal node of $\xi $ is labeled with the number
$1$. Beginning with the first letter, we divide the word $w$ into $%
\left\lfloor n/t\right\rfloor $ blocks with $t$ letters in each and the
suffix of the length $n-t\left\lfloor n/t\right\rfloor $. Let us assume that
the number of nodes labeled with attributes in $\xi $ is less than $%
\left\lfloor n/t\right\rfloor $. Then there is a block such that queries
(attributes) attached to nodes of $\xi $ does not ask about letters from
the block. We replace this block in the word $w$ with the word $w_{0}$ and
denote by $w^{\prime }$ the obtained word. It is clear that $w^{\prime
}\notin L$ and $w^{\prime }\in \Sigma (n,\xi )$, but this is impossible since the
terminal node of the path $\xi $ is labeled with the number 1. Therefore the
depth of $\Gamma $ is greater than or equal to $\left\lfloor n/t\right\rfloor $. Thus, $%
h_{L}^{ma}(n)\geq \left\lfloor n/t\right\rfloor $. It is easy to construct a
decision tree over $L(n)$ that solves the problem of membership for $L(n)$
deterministically and has the depth equals to $n$. Therefore $%
h_{L}^{md}(n)\leq n$. Thus, $H_{L}^{md}(n)=\Theta (n)$ and $%
H_{L}^{ma}(n)=\Theta (n)$.

(b) Let $|L|<\infty $. Then there exists natural $m$ such that $%
L(n)=\emptyset $ for any natural $n\geq m$. Therefore, for each natural $%
n\geq m$, $h_{L}^{md}(n)=0$ and $h_{L}^{ma}(n)=0$. Thus, $H_{L}^{md}(n)=O(1)$
and $H_{L}^{ma}(n)=O(1)$.

Let $L^{C}=\emptyset $, $n$ be a natural number, and $\Gamma $ be a decision
tree over $L(n)$, which consists of the root, a terminal node labeled with $%
1, $ and an edge that leaves the root and enters the terminal node. One can
show that $\Gamma $ solves the problem of membership for $L(n)$
deterministically and has the depth equals to $0$. Therefore $%
h_{L}^{md}(n)=0 $ and $h_{L}^{ma}(n)=0$. Thus, $H_{L}^{md}(n)=O(1)$ and $%
H_{L}^{ma}(n)=O(1)$. 
\end{proof}

\section{Corollaries}

\label{sect4}

In this section, we consider two corollaries of Theorems \ref{T1}--\ref{T3}.

\subsection{Joint Behavior of Functions $H_{L}^{ra}$, $H_{L}^{rd}$, $%
H_{L}^{ma}$, and $H_{L}^{md}$}

\label{sect4.1}

In this section, we assume that each regular factorial language over the alphabet $\Sigma$ is given by
a t-reduced source $I$, which generates the considered language denoted by $%
L_{I}$. To study all possible types of joint behavior of functions $%
H_{L_{I}}^{rd}$, $H_{L_{I}}^{ra}$, $H_{L_{I}}^{md}$, and $%
H_{L_{I}}^{ma}$, we consider five classes of regular factorial languages $%
\mathcal{F}_{1},\ldots ,\mathcal{F}_{5}$ described in the columns 2--4 of
Table \ref{tab1}. In particular, $\mathcal{F}_{1}$ consists of all regular
factorial languages $L_{I}$ for which the source $I$ is an independent
simple source and $cl(I)=0$. It is easy to show that the complexity classes $%
\mathcal{F}_{1},\ldots ,\mathcal{F}_{5}$ are pairwise disjoint, and each
regular factorial language $L_{I}$ belongs to one of these classes. The
behavior of functions $H_{L_{I}}^{rd}$, $H_{L_{I}}^{ra}$, $%
H_{L_{I}}^{md}$, and $H_{L_{I}}^{ma}$ for languages from these classes
is described in the last four columns of Table \ref{tab1}. For each class,
the results considered in Table \ref{tab1} for the functions $%
H_{L_{I}}^{rd}$ and $H_{L_{I}}^{ra}$ follow directly from Theorems \ref%
{T1} and \ref{T2}.

\begin{table}[h]
\caption{Complexity classes $\mathcal{F}_{1},\ldots ,\mathcal{F}_{5}$}
\label{tab1}\center
\begin{tabular}{|l|lll|llll|}
\hline
& $I$ is independent & $cl(I)$ & $L_{I}^{C}$ & $H_{L_{I}}^{rd}$ & $%
H_{L_{I}}^{ra}$ & $H_{L_{I}}^{md}$ & $H_{L_{I}}^{ma}$ \\
& simple source &  &  &  &  &  &  \\ \hline
$\mathcal{F}_{1}$ & Yes & $=0$ &  & $O(1)$ & $O(1)$ & $O(1)$ & $O(1)$ \\
$\mathcal{F}_{2}$ & Yes & $=1$ &  & $O(1)$ & $O(1)$ & $\Theta (n)$ & $\Theta
(n)$ \\
$\mathcal{F}_{3}$ & Yes & $\geq 2$ &  & $\Theta (\log n)$ & $O(1)$ & $\Theta
(n)$ & $\Theta (n)$ \\
$\mathcal{F}_{4}$ & No &  & $\neq \emptyset $ & $\Theta (n)$ & $\Theta (n)$
& $\Theta (n)$ & $\Theta (n)$ \\
$\mathcal{F}_{5}$ & No &  & $=\emptyset $ & $\Theta (n)$ & $\Theta (n)$ & $%
O(1)$ & $O(1)$ \\ \hline
\end{tabular}%
\end{table}

We now consider the behavior of the functions $%
H_{L_{I}}^{md}$ and $H_{L_{I}}^{ma}$ for each of the classes $\mathcal{%
F}_{1},\ldots ,\mathcal{F}_{5}$.
Let $I=(G,q_{0},Q)$ be a t-reduced source over the alphabet $\Sigma$, which generates a regular  factorial language.

Let $L_{I}\in \mathcal{F}_{1}$. Since $cl(I)=0$, $G$ is a directed acyclic
graph, and the language $L_{I}$ is finite. Using Theorem \ref{T3} we obtain $%
H_{L_{I}}^{md}(n)=O(1)$ and $H_{L_{I}}^{ma}(n)=O(1)$.

Let $L_{I}\in \mathcal{F}_{2}$. Since $cl(I)=1$, $G$ is a graph containing a
cycle, and the language $L_{I}$ is infinite. By Lemma 4.2 \cite{Moshkov00}, $%
|L_{I}(n)|=O(1)$. Therefore $L_{I}^{C}\neq \emptyset $. Using Theorem \ref{T3}
we obtain $H_{L_{I}}^{md}(n)=\Theta (n)$ and $H_{L_{I}}^{ma}(n)=\Theta (n)$.

Let $L_{I}\in \mathcal{F}_{3}$. Since $cl(I)\geq 2$, $G$ is a graph
containing a cycle, and the language $L_{I}$ is infinite. By Lemma 4.2 \cite{Moshkov00},
$|L_{I}(n)|=O(n^{cl(I)})$. Therefore $L_{I}^{C}\neq \emptyset $. Using Theorem %
\ref{T3} we obtain $H_{L_{I}}^{md}(n)=\Theta (n)$ and $H_{L_{I}}^{ma}(n)=%
\Theta (n)$.

Let $L_{I}\in \mathcal{F}_{4}$. Since $I$ is not an independent simple
source, $G$ is a graph containing a cycle, and the language $L_{I}$ is
infinite. We know that $L_{I}^{C}\neq \emptyset $. Using Theorem \ref{T3} we
obtain $H_{L_{I}}^{md}(n)=\Theta (n)$ and $H_{L_{I}}^{ma}(n)=\Theta (n)$.

Let $L_{I}\in \mathcal{F}_{5}$. Then $L_{I}^{C}=\emptyset $. Using
Theorem \ref{T3} we obtain $H_{L_{I}}^{md}(n)=O(1)$ and $%
H_{L_{I}}^{ma}(n)=O(1)$.

We now show that the classes $\mathcal{F}_{1},\ldots ,\mathcal{F}_{5}$ are nonempty. For simplicity, we assume that $\Sigma = E$, where $E = \{0,1\}$. It is easy to generalize the considered examples to the case of an arbitrary finite alphabet $\Sigma$ with at least two letters. In the examples of sources, the initial node
is labeled with the symbol $+$, and  all nodes are
terminal.

\begin{figure}[th]
\centering
\includegraphics[width=0.32\textwidth]{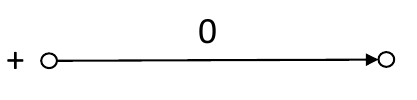}
\caption{Source $I_{1}$}
\label{Fig1}
\end{figure}

 Denote by $I_{1}$ the source over the alphabet $E$ depicted in Fig. \ref{Fig1}. One can show that $%
I_{1}$ is an independent simple t-reduced source and $cl(I_{1})=0$.
This source generates the language $L_{I_{1}}=\{\lambda,0\}$, which is factorial.
Therefore $L_{I_{1}}\in \mathcal{F}_{1}$.

\begin{figure}[h]
\centering
\includegraphics[width=0.17\textwidth]{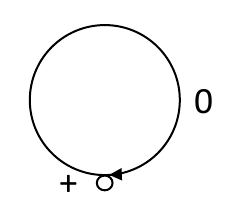}
\caption{Source $I_{2}$}
\label{Fig2}
\end{figure}

Denote by $I_{2}$ the source over the alphabet $E$ depicted in Fig. \ref{Fig2}. One can
show that $I_{2}$ is an independent simple t-reduced source and $%
cl(I_{2})=1$. This source generates the language $L_{I_{2}}=\{0^{i}:i\in
\omega \}$, which is factorial. Therefore $L_{I_{2}}\in $ $\mathcal{F}_{2}$.

\begin{figure}[th]
\centering
\includegraphics[width=0.41\textwidth]{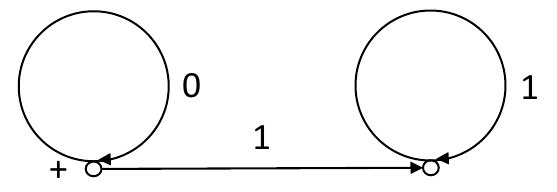}
\caption{Source $I_{3}$}
\label{Fig3}
\end{figure}

 Denote by $I_{3}$ the source over the alphabet $E$ depicted in
Fig. \ref{Fig3}. One can show that $I_{3}$ is an independent simple t-reduced source
and $cl(I_{1})=2$. This source generates the language $L_{I_{3}}=%
\{0^{i}1^{j}:i,j\in \omega \}$, which is factorial. Therefore $L_{I_{3}}\in $
$\mathcal{F}_{3}$.

\begin{figure}[th]
\centering
\includegraphics[width=0.41\textwidth]{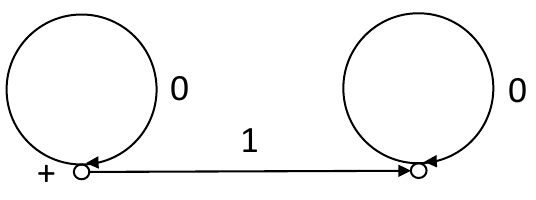}
\caption{Source $I_{4}$}
\label{Fig4}
\end{figure}

Denote by $I_{4}$ the source over the alphabet $E$ depicted in Fig. \ref{Fig4}. One can
show that $I_{4}$ is a dependent simple  t-reduced source generating the language $%
L_{I_{4}}=\{0^{i}1^{j}0^{k}:i,k\in \omega ,j\in \{0,1\}\}$, which is
factorial. It is clear that $L_{I_{4}}^{C}\neq \emptyset $. Therefore $%
L_{I_{4}}\in $ $\mathcal{F}_{4}$.

\begin{figure}[th]
\centering
\includegraphics[width=0.17\textwidth]{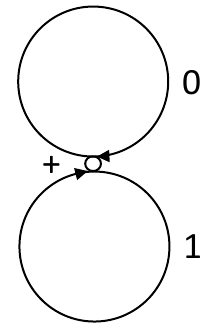}
\caption{Source $I_{5}$}
\label{Fig5}
\end{figure}

 Denote by $I_{5}$ the source over the alphabet $E$
depicted in Fig. \ref{Fig5}. One can show that $I_{5}$ is a t-reduced source that is
not simple. This source generates the language $L_{I_{5}}=E^{\ast }$, which
is factorial. It is clear that $L_{I_{5}}^{C}=\emptyset $. Therefore $%
L_{I_{5}}\in $ $\mathcal{F}_{5}$.

A regular factorial language $L$ can have different t-reduced sources, which
generate it. However, for each of such sources $I$, the language $L_{I}=L$
will belong to the same complexity class. Let us assume the contrary: there
exist a regular factorial language $L$ and two t-reduced sources $I_{1}$
and $I_{2}$, which generate it and for which languages $L_{I_{1}}$ and $%
L_{I_{2}}$ belong to different complexity classes. Then, for some pair $%
bc\in \{rd,ra,md,ma\}$, the functions $H_{L_{I_{1}}}^{bc}$ and $%
H_{L_{I_{2}}}^{bc}$  have different behavior, but this is impossible
since $H_{L_{I_{1}}}^{bc}(n)=H_{L_{I_{2}}}^{bc}(n)$ for any natural $n$.

\subsection{Languages Over Alphabet $\{0,1\}$ Given by One Forbidden Word}

Let $E=\{0,1\}$, $\alpha \in E^{\ast }$, and $\alpha \neq \lambda $. We denote by $L(\alpha )$ the language over the alphabet $E$, which consists of all words from $E^{\ast }$ that does not contain $\alpha $ as a factor. This is  a regular factorial language with $\mathit{MF}(L(\alpha))=\{\alpha \}$.
The following theorem indicates for each nonempty word $\alpha \in E^{\ast }$  the complexity class  $\mathcal{F}_{i}$ to which the language $L(\alpha )$ belongs.

\begin{theorem}
\label{T4} Let $\alpha \in E^{\ast }$ and $\alpha \neq \lambda $.

(a) If $\alpha \in \{0,1\}$, then $L(\alpha )\in \mathcal{F}_{2}$.

(b) If $\alpha \in \{01,10\}$, then $L(\alpha )\in \mathcal{F}_{3}$.

(c) If $\alpha \notin \{0,1,01,10\}$, then $L(\alpha )\in
\mathcal{F}_{4}$.
\end{theorem}

We now describe a t-reduced source $I(\alpha )$ that generates the language $L(\alpha )$ for a  nonempty word $\alpha \in E^{\ast }$.
Let  $\alpha =a_{1}\cdots a_{n}$, $\alpha _{0}=\lambda $, and $%
\alpha _{i}=$ $a_{1}\cdots a_{i}$ for $i=1,\ldots ,n-1$. The set $P(\alpha
)=\{\alpha _{0},\alpha _{1},\ldots ,\alpha _{n-1}\}$ is the set of all
proper prefixes of the word $\alpha $. Then $I(\alpha )=(G,q_{0},Q)$, where
the set of nodes of the graph $G$ is equal to $P(\alpha )$, $q_{0}=\alpha
_{0}$, and $Q=P(\alpha )$. For $i=0,\ldots ,n-2$, an edge leaves the node $%
\alpha _{i}$ and enters the node $\alpha _{i+1}$. This edge is labeled with
the letter $a_{i+1}$. For $i=0,\ldots ,n-1$, an edge leaves the node $\alpha
_{i}$ and enters the node $\alpha _{j}\in P(\alpha )$ such that $\alpha _{j}$
is the longest suffix of the word $\alpha _{i}\bar{a}_{i+1}$, where $\bar{a}%
_{i+1}=0$ if $a_{i+1}=1$ and $\bar{a}_{i+1}=1$ if $a_{i+1}=0$.  This edge is labeled with
the letter $\bar{a}_{i+1}$. It is easy to show that $I(\alpha )$ is a t-reduced source over the alphabet $E$. From Theorem
10 \cite{Crochemore98} it follows that the  source $I(\alpha )$ generates the
language $L(\alpha )$.

Let $\alpha \in E^{\ast }\setminus \{\lambda \}$ and $\alpha =a_{1}\cdots
a_{n}$. We denote by $\bar{\alpha}$  the word $\bar{a}_{1}\cdots \bar{a}_{n}$%
. It is easy to prove the following statement.

\begin{lemma}
\label{L2}Let $\alpha \in E^{\ast }$ and $\alpha \neq \lambda $. Then $H_{L(%
\bar{\alpha})}^{bc}(n)=H_{L(\alpha )}^{bc}(n)$ for any pair $bc\in \{rd,ra,md,$ $ma\}$ and any natural $n$.
\end{lemma}

\begin{lemma}
\label{L3}Let $\alpha \in E^{\ast }\setminus \{\lambda \}$, $\beta \in
E^{\ast }$, and $L(\alpha)\in \mathcal{F}_{4}$. Then $L(\alpha \beta)\in \mathcal{F}_{4}$.
\end{lemma}

\begin{proof}
Since  $L(\alpha)\in \mathcal{F}_{4}$, $H_{L(\alpha )}^{rd}(n)=\Theta
(n)$ and $H_{L(\alpha )}^{ra}(n)=\Theta (n)$.
One can show that $L(\alpha )\subseteq L(\alpha \beta )$. Using this fact it is not difficult to prove that $%
H_{L(\alpha )}^{rd}(n)\leq H_{L(\alpha \beta )}^{rd}(n)$ and $H_{L(\alpha
)}^{ra}(n)\leq H_{L(\alpha \beta )}^{ra}(n)$ for any natural $n$. From here and from
Theorems \ref{T1} and \ref{T2} it follows that $H_{L(\alpha \beta
)}^{rd}(n)=\Theta (n)$ and $H_{L(\alpha \beta )}^{ra}(n)=\Theta (n)$.

Since $\alpha \beta  \notin L(\alpha \beta )$, $L(\alpha \beta )^{C}\neq \emptyset $. The
source $I(\alpha \beta )$ contains at least one circle formed by the edge that
leaves and enters the node $\lambda $ and is labeled with the letter $\bar{a}_{1}$,
where $a_{1}$ is the first letter of the word $\alpha $. Therefore the
language $L(\alpha \beta )$ is infinite. By Theorem \ref{T3},  $%
H_{L(\alpha \beta )}^{md}(n)=\Theta (n)$ and $H_{L(\alpha \beta )}^{ma}(n)=\Theta (n)$. Thus, $L(\alpha \beta)\in \mathcal{F}_{4}$.
\end{proof}

\begin{proof}[of Theorem \ref{T4}]
 In each figure depicting a source $I(\alpha)$, $\alpha \in E^{\ast }\setminus \{\lambda \}$, we label each node with a corresponding prefix of the word $\alpha$.

\begin{figure}[th]
\centering
\includegraphics[width=0.16\textwidth]{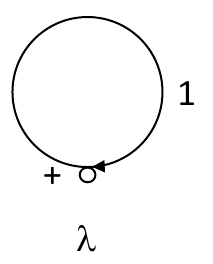}
\caption{Source $I(0)$}
\label{Fig6}
\end{figure}
(a) The source $I(0)$ is depicted in Fig. \ref{Fig6}. This is an independent  simple
t-reduced source with $cl(I(0))=1$. Therefore $L(0)\in \mathcal{F}_{2}$. By Lemma \ref{L2}, $L(1)\in \mathcal{F}%
_{2} $.

\begin{figure}[h]
\centering
\includegraphics[width=0.4\textwidth]{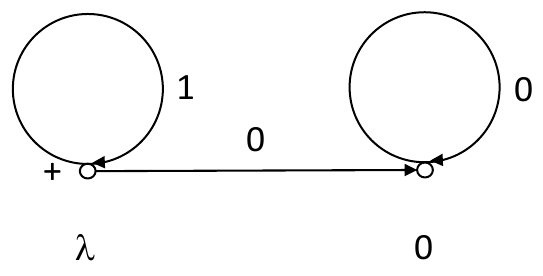}
\caption{Source $I(01)$}
\label{Fig7}
\end{figure}

(b) The source $I(01)$ is depicted in Fig. \ref{Fig7}. This is an independent simple
t-reduced source with $cl(I(01))=2$. Therefore $L(01)\in \mathcal{F}_{3}$. By Lemma \ref{L2}, $L(10)\in
\mathcal{F}_{3}$.

\begin{figure}[h]
\centering
\includegraphics[width=0.33\textwidth]{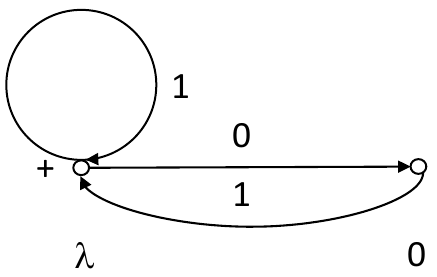}
\caption{Source $I(00)$}
\label{Fig8}
\end{figure}

(c) The source $I(00)$ is depicted in Fig. \ref{Fig8}. This is not a simple source.
It is clear that $L(00)^C\neq \emptyset$. Therefore $L(00)\in
\mathcal{F}_{4}$. By Lemma \ref{L2}, $L(11)\in \mathcal{F}_{4}$. Using
Lemma  \ref{L3} we obtain $L(000),L(001),L(110),L(111)\in
\mathcal{F}_{4}$.

\begin{figure}[th!]
\centering
\includegraphics[width=0.6\textwidth]{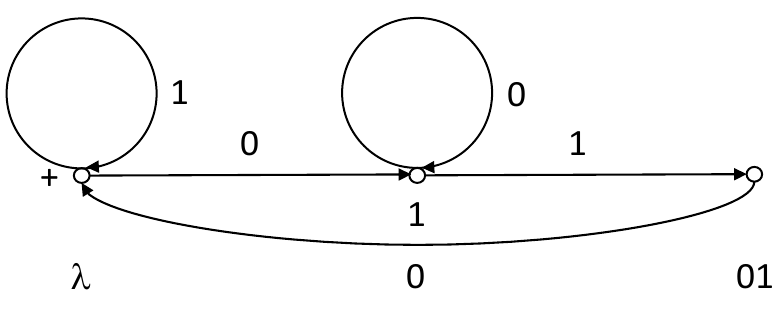}
\caption{Source $I(010)$}
\label{Fig9}
\end{figure}

The source $I(010)$ is depicted in Fig. \ref{Fig9}. This is not a simple source.
 It is clear that $L(010)^C\neq \emptyset$. Therefore $L(010)\in
\mathcal{F}_{4}$. By Lemma \ref{L2}, $L(101)\in \mathcal{F}_{4}$.

\begin{figure}[h]
\centering
\includegraphics[width=0.6\textwidth]{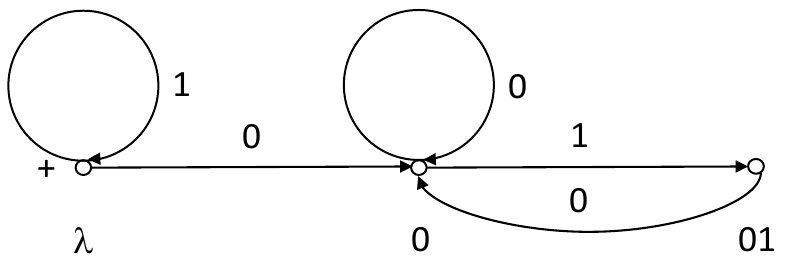}
\caption{Source $I(011)$}
\label{Fig10}
\end{figure}

The source $I(011)$ is depicted in Fig. \ref{Fig10}. This is not a simple source.
It is clear that $L(011)^C\neq \emptyset$. Therefore $L(011)\in
\mathcal{F}_{4}$. By Lemma \ref{L2}, $L(100)\in \mathcal{F}_{4}$.

We proved that, for any word $\alpha \in E^{\ast }$ of the length three, $%
L(\alpha )\in \mathcal{F}_{4}$. Using Lemma  \ref{L3} we
obtain that, for any word $\alpha \in E^{\ast }$ of the length greater than or equal to four, $L(\alpha )\in \mathcal{F}_{4}$. 
\end{proof}

\subsection*{Acknowledgments}
Research reported in this publication was supported by King Abdullah
University of Science and Technology (KAUST).

\bibliographystyle{spmpsci}
\bibliography{fr-languages}

\end{document}